\documentclass[a4paper,UKenglish,cleveref, autoref]{lipics-v2019}

\usepackage{xspace}
\usepackage{smallsec}
\usepackage{hyperref}

\newcommand{\stam}[1]{}
\newcommand{\zug}[1]{\langle #1  \rangle}
\newcommand{\set}[1]{\{ #1 \}}

\newcommand{\Max}{{Max}\xspace}
\newcommand{\Min}{{Min}\xspace}

\newcommand{\RTB}{\texttt{RTB}\xspace}

\newcommand{\Exp}{\mathbb{E}}

\newcommand{\G}{{\cal G}}

\newcommand{\Q}{\mathbb{Q}}
\newcommand{\Qpos}{\Q_{\geq 0}}
\newcommand{\Nat}{\mbox{I$\!$N}}
\newcommand{\Real}{\mbox{I$\!$R}}

\newcommand{\St}{\mbox{St}}

\newcommand{\Pot}{\mbox{Pot}}
\newcommand{\PO}{Player~$1$\xspace}
\newcommand{\PT}{Player~$2$\xspace}
\newcommand{\PLi}{Player~$i$\xspace}

\newcommand{\thresh}{\texttt{Th}\xspace}

\newcommand{\MP}{\texttt{MP}\xspace}

\usepackage{enumitem}
\usepackage[T1]{fontenc}

\title{Bidding Mechanisms in Graph Games} 

\titlerunning{Bidding Mechanisms in Graph Games}

\author{Guy Avni}{IST Austria, Austria}{guy.avni@ist.ac.at}{}{}

\author{Thomas A. Henzinger}{IST Austria, Austria}{tah@ist.ac.at}{}{}

\author{{\DJ}or{\dj}e \v{Z}ikeli\'c}{IST Austria, Austria}{djordje.zikelic@ist.ac.at}{}{}

\authorrunning{G. Avni, T. Henzinger, and \DJ. \v{Z}ikeli\'c}

\Copyright{John Q. Public and Joan R. Public}



\category{}

\relatedversion{}

\supplement{}

\acknowledgements{This research was supported in part by the Austrian Science Fund (FWF) under grants S11402-N23 (RiSE/SHiNE), Z211-N23 (Wittgenstein Award), and M 2369-N33 (Meitner fellowship), and from the European Union's Horizon 2020 research and innovation programme under the Marie Sk\l odowska-Curie Grant Agreement No. 665385.}

\nolinenumbers 

\EventEditors{John Q. Open and Joan R. Access}
\EventNoEds{2}
\EventLongTitle{42nd Conference on Very Important Topics (CVIT 2016)}
\EventShortTitle{CVIT 2016}
\EventAcronym{CVIT}
\EventYear{2016}
\EventDate{December 24--27, 2016}
\EventLocation{Little Whinging, United Kingdom}
\EventLogo{}
\SeriesVolume{42}
\ArticleNo{23}

\begin{document}

\maketitle

\begin{abstract}
In two-player games on graphs, the players move a token through a graph to produce a finite or infinite path, which determines the qualitative winner or quantitative payoff of the game. We study {\em bidding games} in which the players bid for the right to move the token. Several bidding rules were studied previously. In {\em Richman} bidding, in each round, the players simultaneously submit bids, and the higher bidder moves the token and pays the other player. {\em Poorman} bidding is similar except that the winner of the bidding pays the ``bank'' rather than the other player. {\em Taxman} bidding spans the spectrum between Richman and poorman bidding. They are parameterized by a constant $\tau \in [0,1]$: portion $\tau$ of the winning bid is paid to the other player, and portion $1-\tau$ to the bank. While finite-duration (reachability) taxman games have been studied before, we present, for the first time, results on {\em infinite-duration} taxman games. It was previously shown that both Richman and poorman infinite-duration games with qualitative objectives reduce to reachability games, and we show a similar result here. Our most interesting results concern quantitative taxman games, namely {\em mean-payoff} games, where poorman and Richman bidding differ significantly. A central quantity in these games is the {\em ratio} between the two players' initial budgets. While in poorman mean-payoff games, the optimal payoff of a player depends on the initial ratio, in Richman bidding, the payoff depends only on the structure of the game. In both games the optimal payoffs can be found using  (different) probabilistic connections with {\em random-turn based} games in which in each turn, instead of bidding, a coin is tossed to determine which player moves.  While the value with Richman bidding equals the value of a random-turn based game with an un-biased coin, with poorman bidding, the bias in the coin is the initial ratio of the budgets. We give a complete classification of mean-payoff taxman games that is based on a probabilistic connection: the value of a taxman bidding game with parameter $\tau$ and initial ratio $r$, equals the value of a random-turn based game that uses a coin with bias $F(\tau, r) = \frac{r+\tau\cdot (1-r)}{1+\tau}$. Thus, we show that Richman bidding is the exception; namely, for every $\tau <1$, the value of the game depends on the initial ratio. Our proof technique simplifies and unifies the previous proof techniques for both Richman and poorman bidding. 
\end{abstract}

\section{Introduction}
Two-player infinite-duration games on graphs are a central class of games in formal verification \cite{AG11}, where they are used, for example, to solve synthesis \cite{PR89}, and they have deep connections to foundations of logic \cite{Rab69}. A graph game proceeds by placing a token on a vertex in the graph, which the players move throughout the graph to produce an infinite path (``play'') $\pi$. The game is zero-sum and $\pi$ determines the winner or payoff. Graph games can be classified according to the players' objectives. For example, the simplest objective is {\em reachability}, where \PO wins iff an infinite path visits a designated target vertex. Another classification of graph games is the {\em mode of moving} the token. The most studied mode of moving is {\em turn based}, where the players alternate turns in moving the token.

In {\em bidding games}, in each turn, an ``auction'' is held between the two players in order to determine which player moves the token. The bidding mode of moving was introduced in \cite{LLPSU99,LLPU96} for reachability games, where the following bidding rules where defined. In {\em Richman} bidding (named after David Richman), each player has a budget, and before each turn, the players submit bids simultaneously, where a bid is legal if it does not exceed the available budget. The player who bids higher wins the bidding, pays the bid to the other player, and moves the token. A second bidding rule called {\em poorman} bidding in \cite{LLPSU99}, is similar except that the winner of the bidding pays the ``bank'' rather than the other player. Thus, the bid is deducted from his budget and the money is lost. A third bidding rule on which we focus in this paper, called {\em taxman} in \cite{LLPSU99} spans the spectrum between poorman and Richman bidding. Taxman bidding is parameterized by $\tau \in [0,1]$: the winner of a bidding pays portion $\tau$ of his bid to the other player and portion $1-\tau$ to the bank. Taxman bidding with $\tau=1$ coincides with Richman bidding and taxman bidding with $\tau=0$ coincides with poorman bidding. 

Bidding games are relevant for several communities in Computer Science. In formal methods, graph games are used to reason about systems. Poorman bidding games naturally model concurrent systems where processes pay the scheduler for moving. Block-chain technology like Etherium is an example of such a system, which is a challenging to formally verify \cite{CGV18,ABC16}. In Algorithmic Game Theory, auction design is a central research topic that is motivated by the abundance of auctions for online advertisements \cite{Mut09}. Infinite-duration bidding games can model ongoing auctions and can be used to devise bidding strategies for objectives like: ``In the long run, an advertiser's ad should show at least half of the time''. In Artificial Intelligence, bidding games with Richman bidding have been used to reason about combinatorial negotiations \cite{MKT18}. Finally, {\em discrete-bidding} games \cite{DP10}, in which the granularity of the bids is restricted by assuming that the budgets are given using coins, have been studied mostly for recreational games, like bidding chess \cite{BP09}.

Both Richman and poorman infinite-duration games have a surprising, elegant, though different, mathematical structure as we elaborate below. Our study of taxman bidding aims at a better understanding of this structure and at shedding light on the differences between the seemingly similar bidding rules. 

A central quantity in bidding games is the {\em initial ratio} of the players budgets. Formally, assuming that, for $i \in \set{1,2}$, \PLi's initial budget is $B_i$, we say that \PO's initial ratio is $B_1/(B_1+B_2)$. The central question that was studied in \cite{LLPSU99} regards the existence of a necessary and sufficient initial ratio to guarantee winning the game. Formally, the {\em threshold ratio} in a vertex $v$, denoted $\thresh(v)$, is such that if \PO's initial ratio exceeds $\thresh(v)$, he can guarantee winning the game, and if his initial ratio is less than $\thresh(v)$, \PT can guarantee winning the game\footnote{When the initial ratio is exactly $\thresh(v)$, the winner depends on the mechanism with which ties are broken. Our results do not depend on a specific tie-breaking mechanism.Tie-breaking mechanisms are particularly important in discrete-bidding games \cite{AAH19}.}. Existence of threshold ratios in reachability games for all three bidding mechanisms was shown in \cite{LLPSU99}. 

Richman reachability games have an interesting probabilistic connection \cite{LLPU96}. To state the connection, we first need to introduce {\em random-turn based} games. Let $p \in [0,1]$. In a random-turn based game that is parameterized by $p$, in each turn, rather than bidding, the player who moves is chosen by throwing a (possibly) biased coin: with probability $p$, \PO chooses how to move the token, and with probability $1-p$, \PT chooses. Formally, a random-turn based game is a special case of a stochastic game \cite{Con92}. Consider a Richman game $\G$. We construct a ``uniform'' random-turn based game on top of $\G$, denoted $\RTB^{0.5}(\G)$, in which we throw an unbiased coin in each turn. The objective of \PO remains reaching his target vertex. It is well known that each vertex in $\RTB^{0.5}(\G)$ has a {\em value}, which is, informally, the probability of reaching the target when both players play optimally, and which we denote by $val(\RTB^{0.5}(\G), v)$. We are ready to state the probabilistic connection: For every vertex $v$ in the Richman game $\G$, the threshold ratio in $v$ equals $1-val(\RTB(\G), v)$. We note that such a connection is not known and is unlikely to exist in reachability games with neither poorman nor taxman bidding. Random-turn based games have been extensively studied in their own right, mostly with unbiased coin tosses, since the seminal paper \cite{PSSW09}. 

Infinite-duration bidding games have been recently studied with Richman \cite{AHC17} and poorman \cite{AHI18} bidding. For qualitative objectives, namely games in which one player wins and the other player loses, both bidding rules have similar properties. By reducing general qualitative games to reachability games, it is shown that threshold ratios exist for both types of bidding rules. We show a similar result for qualitative games with taxman bidding.

Things get interesting in {\em mean-payoff} games, which are quantitative games: an infinite play has a {\em payoff}, which is \PO's reward and \PT's cost (see an example of a mean-payoff game in Figure~\ref{fig:mean-payoff}). We thus call the players in a mean-payoff game \Max and \Min, respectively. We focus on games that are played on strongly-connected graphs. With Richman bidding \cite{AHC17}, the initial budget of the players does not matter: A mean-payoff game $\G$ has a value $c \in \Real$ that depends only on the structure of the game such that \Min can guarantee a cost of at most $c$ with any positive budget, and with any positive budget, \Max can guarantee a payoff of at least $c-\epsilon$, for every $\epsilon > 0$. Moreover, the value $c$ of $\G$ equals the value of a random-turn based game $\RTB^{0.5}(\G)$ that is constructed on top of $\G$. Since $\G$ is a mean-payoff game, $\RTB^{0.5}(\G)$ is a mean-payoff stochastic game, and its value, which again, is a well-known concept, is the expected payoff when both players play optimally.

\begin{figure}[t]
\begin{minipage}[b]{0.4\linewidth}
\centering
\includegraphics[height=1.5cm]{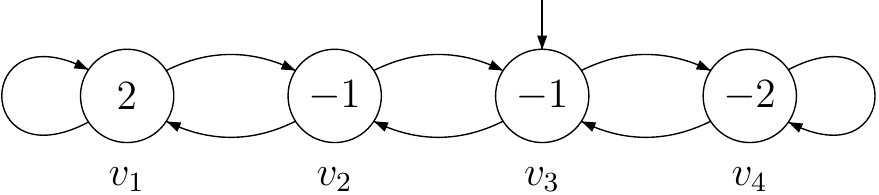}
\end{minipage}
\hspace{0.05\linewidth}
\begin{minipage}[b]{0.51\linewidth}
\centering
\includegraphics[height=3cm]{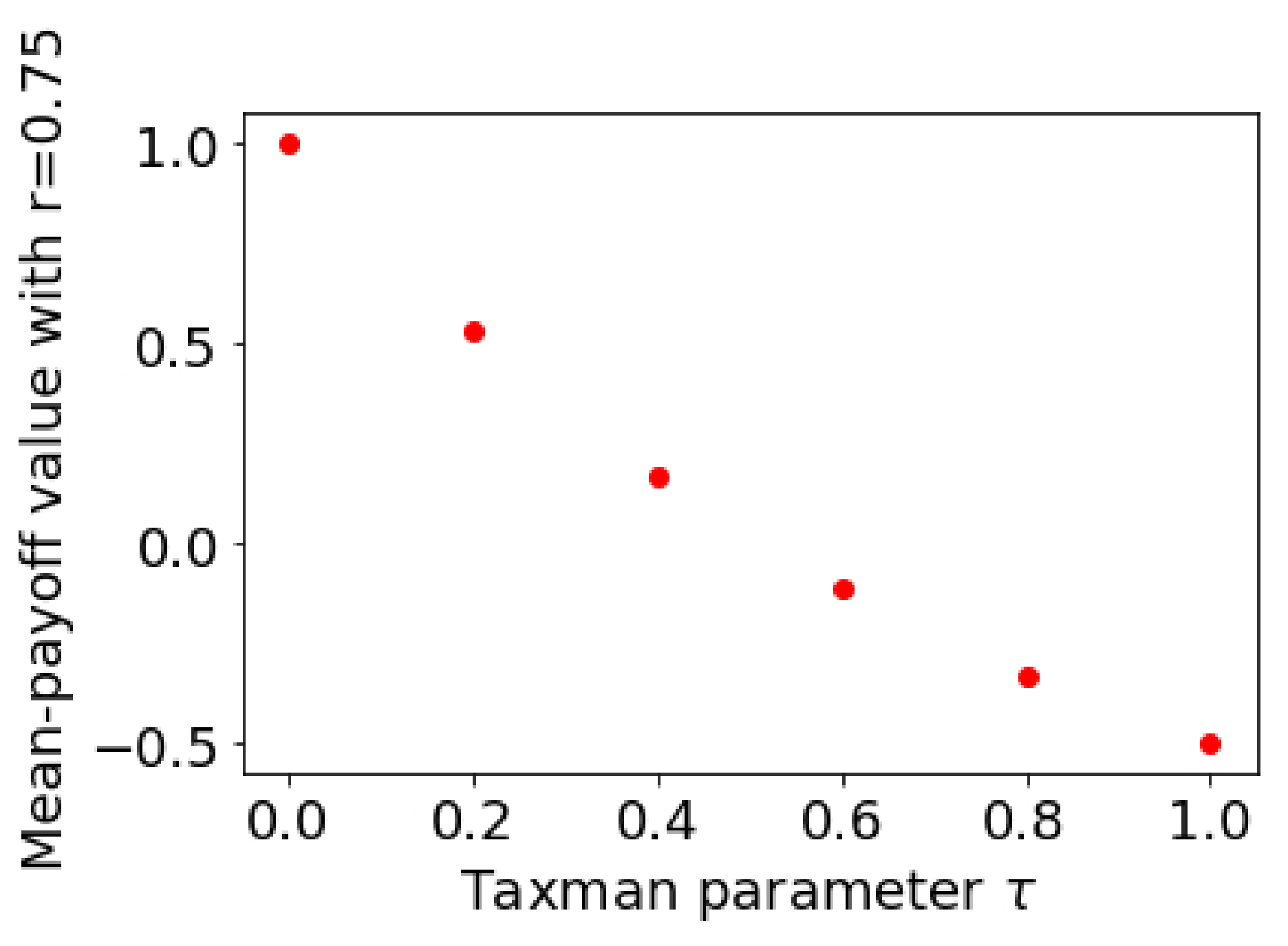}
\label{fig:plot}
\end{minipage}
\caption{On the left, a mean-payoff game $\G$. On the right, the mean-payoff value of $\G$, where the initial ratio is fixed to $0.75$ and the taxman parameter $\tau$ varies. The value of $\G$ with Richman bidding is $-0.5$, with poorman bidding, it is $1$, and, for example, with $\tau = 0.2$, it is $0.533$.}
\label{fig:mean-payoff}
\end{figure}

Poorman mean-payoff games have different properties. Unlike with Richman bidding, the value of the game depends on the initial ratio. That is, with a higher initial ratio, \Max can guarantee a better payoff. More surprisingly, poorman mean-payoff games have a probabilistic connection, which is in fact richer than for Richman bidding. This is surprising since poorman reachability games do not have a probabilistic connection and reachability games tend to be simpler than mean-payoff games. The connection for poorman games is the following: Suppose \Max's initial ratio is $r \in [0,1]$ in a game $\G$. Then, the value in $\G$ with respect to $r$ is the value of the random-turn based game $\RTB^r(\G)$ in which in each turn, we toss a biased coin that chooses \Max with probability $r$ and \Min with probability $1-r$. 

Given this difference between the two bidding rules, one may wonder how do mean-payoff taxman games behave, since these bidding rules span the spectrum between Richman and poorman bidding. Our main contribution is a complete solution to this question: we identify a probabilistic connection for a taxman game $\G$ that depends on the parameter $\tau$ of the bidding and the initial ratio $r$. That is, we show that the value of the game equals the value of the random-turn based game $\RTB^{F(\tau,r)}(\G)$, where $F(\tau,r)=\frac{r+\tau\cdot(1-r)}{1+\tau}$. The construction gives rise to optimal strategies w.r.t. $\tau$ and the initial ratio. As a sanity check, note that for $\tau=1$, we have $F(\tau, r) = 0.5$, which agrees with the result on Richman bidding, and for $\tau=0$, we have $F(\tau, r) = r$, which agrees with the result on poorman bidding. In Figure~\ref{fig:mean-payoff}, we depict some mean-payoff values for a fixed initial ratio and varying taxman parameter. Previous results only give the two endpoints in the plot, and the mid points in the plot are obtained using the results in this paper.

The main technical challenge is constructing an optimal strategy for \Max. The construction involves two components. First, we assign an ``importance'' to each vertex $v$, which we call {\em strength} and denote $\St(v)$. Intuitively, if $\St(v) > \St(u)$, then it is more important for \Max to move in $v$ than in $u$. Second, when the game reaches a vertex $v$, \Max's bid is a careful normalization of $\St(v)$ so that changes in \Max's ratio are matched with the accumulated weights in the game. Finding the right normalization is intricate and it consists of the main technical contribution of this paper. Previous such normalizations were constructed for Richman and poorman mean-payoff games \cite{AHC17,AHI18}. The construction for Richman bidding is much more complicated than the one we present here. The construction for poorman bidding is ad-hoc and does not generalize. Our construction for taxman bidding thus unifies these constructions and simplifies them. It uses techniques that can generalize beyond taxman bidding. Finally, we study, for the first time, complexity problems for taxman games. 

Due to lack of space, some proofs appear in the appendix.

\section{Preliminaries}
A graph game is played on a directed graph $G = \zug{V, E}$, where $V$ is a finite set of vertices and $E \subseteq V \times V$ is a set of edges. The {\em neighbors} of a vertex $v \in V$, denoted $N(v)$, is the set of vertices $\set{u \in V: \zug{v,u} \in E}$. A {\em path} in $G$ is a finite or infinite sequence of vertices $v_1,v_2,\ldots$ such that for every $i \geq 1$, we have $\zug{v_i,v_{i+1}} \in E$. 

\paragraph*{Bidding games} Each \PLi has a budget $B_i \in \Real^{\geq 0}$. In each turn a bidding determines which player moves the token. Both players simultaneously submit bids, where a bid $b_i$ for \PLi is legal if $b_i \leq B_i$. The player who bids higher wins the bidding, where we assume some mechanism to break ties, e.g., always giving \PO the advantage, and our results are not affected by the specific tie-breaking mechanism at use. The winner moves the token and pays his bid, where we consider three bidding mechanisms that differ in where the winning bid is paid. Suppose \PO wins a bidding with his bid of $b$.

\begin{itemize}
\item In {\bf Richman} bidding, the winner pays to the loser, thus the new budgets are $B_1 - b$ and $B_2 + b$.
\item In {\bf poorman} bidding, the winner pays to the bank, thus the new budgets are $B_1 - b$ and $B_2$.
\item In {\bf taxman} bidding with parameter $\tau \in [0,1]$, the winner pays portion $\tau$ to the other player and $(1-\tau)$ to the bank, thus the new budgets are $B_1 - b$ and $B_2 + (1-\tau) \cdot b$. 
\end{itemize}

A central quantity in bidding games is the {\em ratio} of a player's budget from the total budget.
\begin{definition} ({\bf Ratio})
Suppose the budget of \PLi is $B_i$, for $i \in \set{1,2}$, at some point in the game. Then, \PLi's {\em ratio} is $B_i/(B_1 + B_2)$.  The {\em initial ratio} refers to the ratio of the initial budgets, namely the budgets before the game begins. We restrict attention to games in which both players start with positive initial budgets, thus the initial ratio is in $(0,1)$. 
\end{definition}

\paragraph*{Strategies and plays}
A {\em strategy} is a recipe for how to play a game. It is a function that, given a finite {\em history} of the game, prescribes to a player which {\em action} to take,  where we define these two notions below. For example, in turn-based games, a strategy takes as input, the sequence of vertices that were visited so far, and it outputs the next vertex to move to. In bidding games, histories and strategies are more involved as they maintain the information about the bids and winners of the bids. Formally, a history in a bidding game is $\pi = \zug{v_1, b_1, i_1}, \ldots, \zug{v_k, b_k, i_k}, v_{k+1} \in (V \times \Real \times \set{1,2})^*\cdot V$, where for $1 \leq j \leq k+1$, the token is placed on vertex $v_j$ at round $j$, for $1 \leq j \leq k$, the winning bid is $b_j$ and the winner is Player~$i_j$. Consider a finite history $\pi$. For $i \in \set{1,2}$, let $W_i(\pi) \subseteq \set{1,\ldots, k}$ denote the indices in which \PLi is the winner of the bidding in $\pi$. Let $B^I_i$ be the initial budget of \PLi. \PLi's budget following $\pi$, denoted $B_i(\pi)$, depends on the bidding mechanism. For example, in Richman bidding, $B_1(\pi) = B^I_i - \sum_{j \in W_1(\pi)} b_j + \sum_{j \in W_2(\pi)} b_j$, $B_2$ is defined dually, and the definition is similar for taxman and poorman bidding. Given a history $\pi$ that ends in $v$, a strategy for \PLi prescribes an action $\zug{b, v}$, where $b \leq B_i(\pi)$ is a bid that does not exceed the available budget and $v$ is a vertex to move to upon winning, where we require that $v$ is a neighbor of $v_{k+1}$.  An initial vertex, initial budgets, and two strategies for the players determine a unique infinite {\em play} $\pi$ for the game. The vertices that $\pi$ visits form an infinite path $path(\pi)$.

\paragraph*{Objectives} 
An objective $O$ is a set of infinite paths. \PO wins an infinite play $\pi$ iff $path(\pi) \in O$. We call a strategy $f$ {\em winning} for \PO w.r.t. an objective $O$ if for every strategy $g$ of \PT the play that $f$ and $g$ determine is winning for \PO. Winning strategies for \PT are defined dually. We consider the following qualitative objectives:
\begin{enumerate}
    \item In \textit{reachability games}, \PO has a target vertex $t$ and an infinite play is winning iff it visits $t$.
    \item In \textit{parity games}, each vertex is labeled with an index in $\{1,\dots,d\}$. An infinite path is winning for \PO iff the parity of maximal index visited infinitely often is odd.
    \item \textit{Mean-payoff games} are played on weighted directed graphs, with weights given by a function $w:V \rightarrow \Q$. Consider an infinite path $\eta =v_1,v_2,\dots \in V^\omega$. For $n \in \Nat$, the prefix of length $n$ of $\eta$ is $\eta^n$, and we define its {\em energy} to be $E(\eta^n)=\sum_{i=1}^{n} w(v_i)$. The {\em payoff} of $\eta$ is $\MP(\eta) = \liminf_{n\rightarrow \infty}E(\eta^n)/n$. \PO wins $\eta$ iff $\MP(\eta) \geq 0$.
\end{enumerate}

Mean-payoff games are quantitative games. We think of the payoff as \PO's reward and \PT's cost, thus in mean-payoff games, we refer to \PO as \Max and to \PT as \Min. 

\stam{
\begin{definition}
{\bf (Mean-payoff value)} The {\em value} of a mean-payoff game is $c \in \Real$ iff 

\end{definition}
}

\paragraph*{Threshold ratios}
The first question that arrises in the context of bidding games asks what is the necessary and sufficient initial ratio to guarantee an objective. 

\begin{definition} ({\bf Threshold ratios})
Consider a bidding game $\G$, a vertex $v$, an initial ratio $r$, and an objective $O$ for \PO. The threshold ratio in $v$, denoted $\thresh(v)$, is a ratio in $[0,1]$ such that if $r > \thresh(v)$, then \PO has a winning strategy that guarantees that $O$ is satisfied, and if $r < \thresh(v)$, then \PT has a winning strategy that violates $O$.
\end{definition}

\paragraph*{Random turn-based games}
A {\em stochastic game} \cite{Con92} is a graph game in which the vertices are partitioned between two players and a {\em nature} player. As in turn-based games, whenever the game reaches a vertex that is controlled by \PLi, for $i =1,2$, he choses how the game proceeds, and whenever the game reaches a vertex $v$ that is controlled by nature, the next vertex is chosen according to a probability distribution that depends only on $v$. 

Consider a bidding game $\G$ that is played on a graph $\zug{V, E}$. The {\em random-turn based game} with ratio $r \in [0,1]$ that is associated with $\G$ is a stochastic game that intuitively simulates the following process. In each turn we throw a biased coin that turns heads with probability $r$ and tails with probability $1-r$. If the coin turns heads, then \PO moves the token, and otherwise \PT moves the token. Formally, we define $\RTB^r(\G) = \zug{V_1, V_2, V_N, E, \Pr}$, where each vertex in $V$ is split into three vertices, each controlled by a different player, thus for $\alpha \in \set{1, 2, N}$, we have $V_\alpha = \set{v_\alpha: v \in V}$, nature vertices simulate the fact that \PO chooses the next move with probability $r$, thus $\Pr[v_N,v_1] = r = 1- \Pr[v_N,v_2]$, and reaching a vertex that is controlled by one of the two players means that he chooses the next move, thus $E = \set{\zug{v_\alpha, u_N}: \zug{v,u} \in E \text{ and } \alpha \in \set{1, 2}}$. When $\G$ is a mean-payoff game, the vertices are weighted and we define the weights of $v_1, v_2$, and $v_N$ to be equal to the weight of $v$. 

The following definitions are standard, and we refer the reader to \cite{Put05} for more details. A strategy in a stochastic game is similar to a turn-based game; namely, given the history of vertices visited so far, the strategy chooses the next vertex. Fixing two such strategies $f$ and $g$ for both players gives rise to a distribution $D(f,g)$ on infinite paths. Intuitively, \PO's goal is to maximize the probability that his objective is met. An {\em optimal} strategy for \PO guarantees that the objective is met with probability at least $c$ and, intuitively, he cannot do better, thus \PT has a strategy that guarantees that the objective is violated with probability at least $(1-c)$. It is well known that optimal positional strategies exist for the objectives that we consider.

\begin{definition}
({\bf Values in stochastic games})
Consider a bidding game $\G$, let $r \in [0,1]$, and consider two optimal strategies $f$ and $g$ for the two players in $\RTB^r(\G)$. When $\G$ is a qualitative game with objective $O$, the {\em value} of $\RTB^r(\G)$, denoted $val(\RTB^r(\G))$, is $\Pr_{\eta \sim D(f,g)} \Pr[\eta \in O]$. When $\G$ is a mean-payoff game, the {\em mean-payoff value} of $\RTB^r(\G)$, denoted $\MP(\RTB^r(\G))$, is $\Exp_{\eta \in D(f,g)} \MP(\eta)$.
\end{definition}

\section{Qualitative Taxman Games}
We start by describing the results on reachability bidding games.

\begin{theorem}
\label{thm:reach}
\cite{LLPSU99}
Consider a reachability bidding game $\G$ and a vertex $v$. The threshold ratio exists in $v$ with Richman, poorman, and taxman bidding. Moreover, threshold ratios have the following properties. For the target vertex $t$ of \PO, we have $\thresh(t) = 0$. For every vertex $v$ from which there is no path to $t$, we have $\thresh(v)=1$. Consider some other vertex $v$ and denote $v^+,v^- \in N(v)$ the vertices for which $\thresh(v^-) \leq \thresh(u) \leq \thresh(v^+)$, for every $u \in N(v)$. 
\begin{itemize}
\item In Richman bidding, we have $\thresh(v) = \frac{1}{2}\big(\thresh(v^+) + \thresh(v^-)\big)$. Moreover, $\thresh(v)$ is a rational number and satisfies $\thresh(v) = 1-val(\RTB(\G), v)$. 
\item In poorman bidding, we have $\thresh(v) = \thresh(v^+)/(1+\thresh(v^+) -\thresh(v^-))$.
\item In taxman bidding with parameter $\tau$, we have $\thresh(v) = \big(\thresh(v^-)+\thresh(v^+)-\tau \cdot \thresh(v^-)\big)/\big(2-\tau \cdot (1+\thresh(v^-)-\thresh(v^+))\big)$.
\end{itemize}
\end{theorem}

It is shown in \cite{AHC17} and \cite{AHI18} that parity games with Richman and poorman bidding reduce to reachability games. We show a similar result for taxman games. The crucial step is the following lemma whose proof can be found in Appendix~\ref{app:SCC-qual}.

\begin{lemma}
\label{lem:SCC-qual}
Consider a taxman reachability game $\G$ that is played on the graph $\zug{V,E}$. Suppose that every vertex in $\G$ has a path to the target of \PO. Then,  for any taxman parameter $\tau$ and every $v \in V$, we have $\thresh(v) = 0$. That is, \PO wins from $v$ with any positive initial budget.
\end{lemma}
\begin{proof}
Let $n=|V|-1$ and $t \in V$ be \PO's target. Suppose the game starts from a vertex $v$, and let $\epsilon>0$ be the initial budget of \PO. Since there is a path from $v$ to \PO's target, there is a path of length at most $n$. Thus, if \PO wins $n$ consecutive biddings, he wins the game. Intuitively, \PO carefully chooses $n$ increasing bids such that if \PT wins one of these bids, \PO's ratio increases by a constant over his initial budget. By repeatedly playing according to such a strategy, \PO guarantees that his ratio increases and will eventually allow him to win $n$ biddings in a row. Formally, if $\tau=0$, then $\G$ is a Richman game and the proof of the lemma can be found in \cite{AHC17}. Otherwise, pick a sufficiently large $r\in \Nat$ such that $\tau>\frac{2}{r-1}$ and $r\geq 3$. Fix $0<m<\frac{\epsilon}{r^n}$. \PO proceeds as follows: after winning $i$ times, for $0\leq i$, he bids $m \cdot r^i$ and, upon winning the bidding, he moves towards $t$ along any shortest path. Since $m+mr+\dots+mr^{n-1}<mr^n<\epsilon$, \PO has sufficient budget to win $n$ consecutive biddings. If \PT does not win any of the first $n$ biddings, \PO wins the game. On the other hand, if \PT wins the $k$-th bidding with $1\leq k \leq n$, we show in Appendix~\ref{app:SCC-qual} that his ratio increases by a fixed amount $b=\frac{mr}{(1-\epsilon)(r-1)}>0$.
\end{proof}

Lemma~\ref{lem:SCC-qual} gives rise to simple reduction from parity taxman games to taxman reachability games. 

\begin{theorem}
\label{thm:parity}
Parity taxman games are linearly reducible to taxman reachability games. Specifically, threshold ratios exist in parity taxman games.
\end{theorem}
\begin{proof}
A {\em bottom strongly-connected component} (BSCC, for short) in $\G$ is a maximal subset of vertices such that every two vertices have a path between them and no edges leave the set. Lemma~\ref{lem:SCC-qual} ensures that when the game is in a BSCC, with any positive initial budget, a player can force the game to reach any other vertex. A strategy that ensures infinitely many visits to a vertex $t$ splits a player's budget into infinitely many positive parts and uses the $i$-th part to force the game to visit $t$ for the $i$-th time. Thus, a BSCC in which the highest parity index is odd is ``winning'' for \PO and these in which the highest parity index is odd are ``losing'' for \PO. We then construct a reachability game by removing the BSCCs of the game and playing a reachability game on the rest of the game, where \PO's targets are his winning BSCCs.
\end{proof}

\section{Mean-Payoff Taxman Games}
This section consists of our main technical contribution. We start by showing a complete classification of the value in strongly-connected mean-payoff taxman games depending on the taxman parameter $\tau$ and the initial ratio. We then extend the solution to general games, where the solution to strongly-connected games constitutes the main ingredient in the solution of the general case.

\subsection{Strongly-Connected Mean-Payoff Taxman Games}
We start by formally defining the value of a strongly-connected mean-payoff game. Lemma~\ref{lem:SCC-qual} implies that in a strongly-connected game, a player can draw the game from every vertex to any other vertex with any positive initial budget. Since mean-payoff objectives are prefix independent, it follows that the vertex from which the game starts does not matter. Indeed, if the game starts at a vertex $v$ with \Max having initial ratio $r+ \epsilon$, then \Max can use $\epsilon/2$ of his budget to draw the game to a vertex $u$ and continue as if he starts the game with initial ratio $r+\epsilon/2$.

\begin{definition}
{\bf (Mean-payoff value)} Consider a strongly-connected mean-payoff game $\G = \zug{V, E, w}$ and a ratio $r \in (0,1)$ and a taxman parameter $\tau \in [0,1]$. The mean-payoff value of $\G$ w.r.t. $r$ and $\tau$, is a value $c \in \Real$ such that for every $\epsilon >0$
\begin{itemize}
\item if \Min's initial ratio is greater than $(1-r)$, then he has a strategy that guarantees that the payoff is at most $c+\epsilon$, and
\item if \Max's initial ratio is greater than $r$, then he has a strategy that guarantees that the payoff is greater than $c-\epsilon$. 
\end{itemize}
\end{definition}

The following theorem, which we prove in the next two sections, summarizes the properties of mean-payoff taxman games. 

\begin{theorem}\label{thm:SCC-MP}
Consider a strongly-connected mean-payoff taxman game $\G$ with taxman parameter $\tau \in [0,1]$ and an initial ratio $r \in (0,1)$. The value of $\G$ w.r.t. $\tau$ and $r$ equals the value of the random-turn based game $\RTB^{F(\tau,r)}(\G)$ in which \Max is chosen to move with probability $F(\tau,r)$ and \Min with probability $1-F(\tau, r)$, where $F(\tau, r) = \frac{r+\tau(1-r)}{1+\tau}$. 
\end{theorem}

We show that in order to prove Theorem~\ref{thm:SCC-MP}, it suffices to prove the following intermediate lemma.

\begin{lemma}
\label{lem:intermediate}
Consider a strongly-connected mean-payoff taxman game $\G$, a taxman parameter $\tau$, and an initial ratio $r \in (0,1)$ such that $\MP(\RTB^{F(\tau,r)}) = 0$ for $F(\tau, r) = \frac{r+\tau(1-r)}{1+\tau}$. Then, for every $\epsilon>0$ \Max has a strategy that guarantees that no matter how \Min plays, the payoff is greater than $-\epsilon$.
\end{lemma}
\begin{proof}[Proof that Lemma~\ref{lem:intermediate} implies Theorem~\ref{thm:SCC-MP}]
First, we may assume that $\MP(\RTB^{F(\tau,r)}) = 0$ since we can decrease all weights by $\MP(\RTB^{F(\tau,r)})$. Recall that the definition of the payoff of an infinite play $\pi = v_1,v_2,\ldots$ is $\lim \inf_{n\to \infty} \frac{1}{n} \sum_{i=1}^n w(v_i)$. Note that since the definition uses $\lim\inf$, it gives \Min an advantage. Constructing a strategy for \Max is thus more challenging and it implies a strategy for \Min as follows. Let $\G'$ be a mean-payoff game that is obtained from $\G$ by multiplying all the weights by $-1$, and associate \Min in $\G$ with \Max in $\G'$ and vice-versa. Observe that $\MP(\RTB^{1-\frac{r+\tau(1-r)}{1+\tau}}(\mathcal{G'}))=-\MP(\RTB^{\frac{r+\tau(1-r)}{1+\tau}}(\mathcal{G}))=0$. Thus, using a strategy for \Max in $\G'$ that guarantees  a payoff that is greater than $-\epsilon$ can be used by \Min to guarantee a payoff in $\G$ that is smaller than $\epsilon$.
\end{proof}

\subsection{The importance of moving} 
The first part of the construction of an optimal strategy for \Max as in Lemma~\ref{lem:intermediate} is to assign, to each vertex $v \in V$, a {\em strength}, denoted $\St(v)$, where $\St(v) \in \Qpos$. Intuitively, if $\St(v) > \St(u)$, for $u,v \in V$, it is more important for \Max to move in $v$ than it is in $u$. We follow the construction in \cite{AHI18}, which uses the concept of {\em potentials}, which is a well-known concept in stochastic game (see \cite{Put05}) and was originally defined in the context of the strategy iteration algorithm \cite{How60}. For completeness, we present the definitions below.

Consider a strongly-connected mean-payoff game $\G$, and let $p \in [0,1]$. Let $f$ and $g$ be two optimal positional strategies in $\RTB^{p}(\G)$, for \Min and \Max, respectively. For a vertex $v \in V$, let $v^-,v^+ \in V$ be such that \Max proceeds from $v$ to $v^+$ according to $g$ and \Min proceeds from $v$ to $v^-$ according to $f$. It is not hard to see that the mean-payoff value in all vertices in $\RTB^p(\G)$ is the same and we denote it by $\MP(\RTB^p(\G))$. We denote the  potential of $v$ by $\Pot^p(v)$ and the strength of $v$ by $\St^p(v)$, and we define them as follows.
\begin{equation*}
\begin{split}
&\Pot^p(v) = p \cdot \Pot^p(v^+) + (1-p) \cdot \Pot^p(v^-) + w(v) - \MP(\RTB^p(\G))  \text{ and } \\
&\St^p(v) = p\cdot (1-p) \cdot \big(\Pot^p(v^+) - \Pot^p(v^-)\big)
\end{split}
\end{equation*}
There are optimal strategies for which $\Pot^p(v^-) \leq \Pot^p(v') \leq \Pot^p(v^+)$, for every $v' \in N(v)$, which can be found, for example, using the strategy iteration algorithm. Note that $\St(v) \geq 0$, for every $v \in V$. 

Consider a finite path $\pi = v_1,\ldots, v_n$ in $\G$. We intuitively think of $\pi$ as a play, where for every $1 \leq i < n$, the bid of \Max in $v_i$ is $\St(v_i)$ and he moves to $v_i^+$ upon winning. Thus, if $v_{i+1} = v_i^+$, we say that \Max won in $v_i$, and if $v_{i+1} \neq v_i^+$, we say that \Max lost in $v_i$. Let $W(\pi)$ and $L(\pi)$ respectively be the indices in which \Max wins and loses in $\pi$. We call \Max wins {\em investments} and \Max loses {\em gains}, where intuitively he {\em invests} in increasing the energy and {\em gains} a higher ratio of the budget whenever the energy decreases. Let $G(\pi)$ and $I(\pi)$ be the sum of gains and investments in $\pi$, respectively, thus $G(\pi) = \sum_{i \in L(\pi)} \St(v_i)$ and $I(\pi) =  \sum_{i \in W(\pi)} \St(v_i)$. Recall that the energy of $\pi$ is $E(\pi) = \sum_{1 \leq i <n} w(v_i)$. The following lemma, which generalizes a similar lemma in \cite{AHI18}, connects the strength with the change in energy. 

\begin{lemma}
\label{lem:magic}
Consider a strongly-connected game $\G$, and $p \in [0,1]$. For every finite path $\pi=v_1,\ldots,v_n$ in $\G$, we have $\Pot^p(v_1) - \Pot^p(v_n) + (n-1)\cdot\MP(\RTB^p(\G)) \leq E(\pi) + G(\pi)/(1-p) - I(\pi)/p$. In particular, when $p = \nu/(\mu+\nu)$ for $\nu,\mu > 0$, there is a constant $P = \min_v \Pot^p(v) - \max_v \Pot^p(v)$ such that $\frac{\nu\cdot \mu}{\nu+\mu}\cdot \big(E(\pi) -P-(n-1)\cdot \MP(\RTB^{\frac{\nu}{\mu+\nu}}(\G))\big) \geq \mu\cdot I(\pi)-\nu\cdot G(\pi)$. 
\end{lemma}

\begin{proof}
The proof is by induction on the length of $\pi$. For $n=1$, the claim is trivial since both sides of the equation are $0$. Suppose the claim is true for all paths of length $n-1$ and we prove it for a path $\pi=v_1,\ldots,v_{n+1}$ of length\footnote{The weight of the last vertex does not participate in the energy calculation, thus the length of a path that traverses $n+1$ vertices has length $n$.} $n$. We consider the case when Max wins in $v_1$ thus $v_2=v_1^+$. The case when Min wins in $v_1$ is proved similarly. Let $\pi'$ be the part of path $\pi$ starting in $v_2$. Since Max wins the first bidding, we have $G(\pi')=G(\pi)$, $I(\pi')=I(\pi)+\St^p(v)$. Hence, by induction hypothesis we have
\begin{equation*}
    \begin{split}
        E(\pi)&+G(\pi)/(1-p)-I(\pi)/p \geq E(\pi')+G(\pi')/(1-p)-I(\pi')/p+w(v_1)-\St^p(v_1)/p \\ 
        &\geq \Pot^p(v_1^+)-\Pot^p(v_{n+1})+(n-1)\cdot\MP(\RTB^p(\G))+w(v_1)-\St^p(v_1)/p \\
        &= \Pot^p(v_1^+)-\Pot^p(v_{n+1})+(n-1)\cdot\MP(\RTB^p(\G))+w(v_1)-(1-p)\cdot (\Pot^p(v_1^+)-\Pot^p(v_1^-)) \\
        &= p\cdot\Pot^p(v_1^+)+(1-p)\cdot\Pot^p(v_1^-)+w(v_1)-\Pot^p(v_{n+1})+(n-1)\cdot\MP(\RTB^p(\G)) \\
        &= \Pot^p(v_1)-\Pot^p(v_{n+1})+n\cdot\MP(\RTB^p(\G)).
    \end{split}
\end{equation*}
\end{proof}

\subsection{Normalizing the bids}
Once we have figured out how important each vertex is, the second challenge in the construction of \Max's strategy is to wisely use his budget such that the changes in the ratios between the players' budgets coincides with the changes in the accumulated energy. Intuitively, Lemma~\ref{lem:ratios} below gives us a recipe to normalize the bids: whenever we reach a vertex $v$, \Max bids $r\cdot (1-r) \cdot \St(v) \cdot \beta_x$, where $\beta_x$ is the normalization factor and $x \in \Real_{\geq 1}$ ties between changes in energy and changes in \Max's ratio, as elaborated after the lemma.

\begin{lemma}
\label{lem:ratios}
Consider a game $\G$, a finite set of non-negative strengths $S \subseteq \Real_{\geq 0}$, a ratio $r \in (0,1)$, and a taxman parameter $\tau \in [0,1]$. For every $K > \frac{\tau r^2+r(1-r)}{\tau (1-r)^2+r(1-r)}$ there exist sequences $(r_x)_{x \geq 1}$ and $(\beta_x)_{x\geq 1}$ with the following properties.
\begin{enumerate}
\item\label{legal} \Max's bid does not exceed his budget, thus, for each position $x \in \Real_{\geq 1}$ and strength $s \in S$, we have $\beta_x \cdot s\cdot r\cdot (r-1) < r_x$.
\item\label{win} \Min cannot force the game beyond position $1$, thus for every $s \in S\backslash \{0\}$ and $1\leq x< 1+rs$, we have $\beta_x \cdot s\cdot r\cdot (r-1) > 1-r_x$.
\item\label{decrease} The ratios tend to $r$ from above, thus for every $x \in \Real_{\geq 1}$, we have $r_x \geq r$, and $\lim_{x \to \infty} r_x = r$.
\item\label{invariant} No matter who wins a bidding, \Max's ratio can only improve. Thus, in case of winning and in case of losing, we respectively have 
\[\frac{r_x - \beta_x \cdot s\cdot r\cdot (r-1)}{1-(1-\tau) \cdot \beta_x \cdot s\cdot r\cdot (r-1)} \geq r_{x+(1-r)\cdot K\cdot s} \text{ and } \frac{r_x + \tau\cdot \beta_x \cdot s\cdot r\cdot (r-1)}{1-(1-\tau) \cdot \beta_x \cdot s\cdot r\cdot (r-1)} \geq r_{x-s\cdot r}\]
\end{enumerate}
\end{lemma}

We first show how Lemma~\ref{lem:ratios} implies Theorem~\ref{thm:SCC-MP}.
\begin{proof}[Proof that Lemma~\ref{lem:ratios} implies Lemma~\ref{lem:intermediate}]
Fix $\epsilon>0$, we construct strategy for Max guaranteeing a payoff greater than $-\epsilon$, as wanted. Observe that
\begin{equation*}
    \frac{r}{r+(1-r)\frac{\tau r^2+r(1-r)}{\tau (1-r)^2 +r(1-r)}} = \frac{r(\tau (1-r)+r)}{\tau r(1-r)+r^2+\tau r^2+r(1-r)} = \frac{r+\tau(1-r)}{1+\tau} = F(\tau,r).
\end{equation*}
Thus, since by assumption $\MP(\RTB^{F(\tau,r)}(\G))=0$ and $\MP(\RTB^p(\G))$ is a continuous function in $p\in [0,1]$ \cite{Cha12,Sol03}, we can pick $K>F(\tau,r)$ such that $\MP(\RTB^{\frac{r}{r+(1-r)K}}(\G))>-\epsilon$. \\
We now describe Max's strategy. We think of the change in \Max's ratio as a walk on $\Real_{\geq 1}$. Each position $x \in \Real_{\geq 1}$ is associated with a ratio $r_x$. The walk starts in a position $x_0$ such that \Max's initial ratio is at least $r_{x_0}$. Let $\nu = r$ and $\mu = K(1-r)$. Suppose the token is placed on a vertex $v \in V$. Then, \Max's bid is $r\cdot(1-r)\cdot \beta_x \cdot \St(v)$ (when ratios of Max and Min are normalized to sum up to $1$) and he proceeds to $v^+$ upon winning. If \Max wins, the walk proceeds up $\mu \cdot \St(v)$ steps to $x+ \mu\St(v)$, and if he loses, the walk proceeds down to $x-\nu\St(v)$. Suppose \Min fixes some strategy and let $\pi=v_1,\ldots,v_n$ be a finite prefix of the play that is generated by the two strategies. Suppose the walk following $\pi$ reaches $x \in \Real$. Then, using the terminology of the previous section, we have $x = x_0 - G(\pi) \cdot \nu + I(\pi) \cdot \mu$. Lemma~\ref{lem:ratios} shows that the walk always stays above $1$, thus $x \geq 1$. Combining with Lemma~\ref{lem:magic}, we get $\frac{\nu+\mu}{\nu \cdot \mu} (1-x_0)+ P+ (n-1)\cdot\MP(\RTB^{\frac{\nu}{\nu+\mu}}(\G)) \leq E(\pi)$. Thus, dividing both sides by $n$ and letting $n\rightarrow \infty$, since $x_0$ and $P$ are constants depending only on $K$ we conclude that this strategy guarantees payoff at least $\MP(\RTB^{\frac{\nu}{\nu+\mu}}(\G))>-\epsilon$.
\end{proof}

We continue to prove Lemma~\ref{lem:ratios}.
\begin{proof}[Proof of Lemma~\ref{lem:ratios}]
Note that $\frac{\tau r^2+r(1-r)}{\tau (1-r)^2 +r(1-r)}$ is well-defined for $r\in (0,1)$. Fix $\tau\in[0,1]$ and $r \in (0,1)$. Let $K > \frac{\tau r^2+r(1-r)}{\tau (1-r)^2 +r(1-r)}$. Observe that the two inequalities in Point~\ref{invariant} are equivalent to:
\begin{equation*}
\begin{split}
    r_{x-rs}-r_x &\leq \tau r(1-r)\beta_x s+(1-\tau)r(1-r)\beta_x s r_{x-rs}, \\
    r_x-r_{x+K(1-r)s} &\geq r(1-r)\beta_xs-(1-\tau)r(1-r)\beta_xsr_{x+K(1-r)s}.
\end{split}
\end{equation*}
Point~\ref{decrease} combined with monotonicity in the above expressions, implies that we can replace the last term in each of them by $r$ in order to obtain stronger inequalities. Therefore, it suffices for $(r_x)_{x\geq 1}$ and $(\beta_x)_{x\geq 1}$ to satisfy
\begin{equation*}
\begin{split}
    r_{x-rs}-r_x &\leq \tau r(1-r)\beta_x s+(1-\tau)r(1-r)\beta_x s r, \\
    r_x-r_{x+K(1-r)s} &\geq r(1-r)\beta_xs-(1-\tau)r(1-r)\beta_xsr,
\end{split}
\end{equation*}
which is equivalent to 
\begin{equation} \label{eq:C}
\begin{split}
    r_{x-rs}-r_x &\leq r(1-r)\beta_x s[\tau +(1-\tau)r], \\
    r_x-r_{x+K(1-r)s} &\geq r(1-r)\beta_x s[1-(1-\tau)r].
\end{split}
\end{equation}
We seek $(r_x)_{x\geq 1}$ and $(\beta_x)_{x\geq 1}$ in the form $r_x=\gamma^{x-1}+(1-\gamma^{x-1})r$ and $\beta_x=\beta\gamma^{x-1}$ for some $\gamma,\beta\in (0,1)$. Note that this choice ensures Points~\ref{legal} and~\ref{decrease}. 
Therefore, we just need to show that we can find $\gamma,\beta\in [0,1]$ for which the inequalities in \eqref{eq:C} hold for any $s \in S$. 
Substituting $r_x$ and $\beta_x$ in terms of $\gamma$ and $\beta$, the inequalities in \eqref{eq:C} reduce to
\begin{equation*}
\begin{split}
    r_{x-rs}-r_x &= \gamma^{x-1}(\gamma^{-rs}-1)(1-r) \stackrel{?}{\leq} \beta\gamma^{x-1}r(1-r)s[\tau+(1-\tau)r], \\
    r_x-r_{x+K(1-r)s} &= \gamma^{x-1}(1-\gamma^{K(1-r)s})(1-r) \stackrel{?}{\geq} \beta\gamma^{x-1}r(1-r)s[1-(1-\tau)r].
\end{split}
\end{equation*}
First, when $s=0$, both sides of both inequalities are equal to $0$ so both inequalities clearly hold. Recall that $S$ is a finite set of non-negative strengths. Thus, when $s > 0$, it takes values in $0<s_1 \leq \ldots \leq s_n$, and the above inequalities are equivalent to
\begin{equation} \label{eq:gamma}
\begin{split}
    \gamma &\geq \big( 1+\beta rs[\tau+(1-\tau)r] \big)^{-\frac{1}{rs}}, \\
    \gamma &\leq \big( 1-\beta rs[1-(1-\tau)r] \big)^{\frac{1}{K(1-r)s}}.
\end{split}
\end{equation}
Since both of these expressions are in $(0,1)$, to conclude that $\gamma,\beta\in(0,1)$ exist, it suffices to show that there is some $\beta\in(0,1)$ such that

\begin{equation}\label{eq:maxmin}
    \max_{s\in \{s_1,\dots,s_n\}}\big( 1+\beta rs[\tau+(1-\tau)r] \big)^{-\frac{1}{rs}} \leq \min_{s\in \{s_1,\dots,s_n\}}\big( 1-\beta rs[1-(1-\tau)r] \big)^{\frac{1}{K(1-r)s}}.
\end{equation}

\noindent Note that the LHS of \eqref{eq:maxmin} is monotonically increasing in $s>0$ whereas the RHS is monotonically decreasing in $s>0$, therefore it suffices to find $\beta\in(0,1)$ for which

\begin{equation}\label{eq:reduce}
    \big( 1+\beta rs_n[\tau+(1-\tau)r] \big)^{-\frac{1}{rs_n}} \leq \big( 1-\beta rs_1[1-(1-\tau)r] \big)^{\frac{1}{K(1-r)s_1}}.
\end{equation}

\noindent By Taylor's theorem $(1+y)^\alpha=1+\alpha y+O(y^2)$, so Taylor expanding both sides of \eqref{eq:reduce} in $\beta>0$ we get 
\begin{equation*}
\begin{split}
    &\big( 1+\beta rs_n[\tau+(1-\tau)r] \big)^{-\frac{1}{rs_n}} = 1-\beta [\tau+(1-\tau)r] +O(\beta^2), \\
    &\big( 1-\beta rs_1[1-(1-\tau)r] \big)^{\frac{1}{K(1-r)s_1}} = 1-\beta \frac{r}{K(1-r)}[1-(1-\tau)r] +O(\beta^2).
\end{split}
\end{equation*}
Therefore, if we show that $[\tau+(1-\tau)r]>\frac{r}{K(1-r)}[1-(1-\tau)r]$, the linear coefficient of $\beta$ on the LHS of \eqref{eq:reduce} will be strictly smaller than the linear coefficient of $\beta$ on the RHS. Thus, for sufficiently small $\beta>0$, \eqref{eq:reduce} will hold, which concludes the proof of the lemma. This condition is equivalent to
\begin{equation*}
   K > \frac{r[1-(1-\tau)r]}{(1-r)[\tau+(1-\tau)r]} = \frac{r[\tau r+(1-r)]}{(1-r)[\tau(1-r)+r]} = \frac{\tau r^2+r(1-r)}{\tau (1-r)^2 +r(1-r)},
\end{equation*}
which is true by assumption. Thus, Points~\ref{legal},~\ref{decrease}, and~\ref{invariant} hold. In Appendix~\ref{app:point 2}, we show that Point~\ref{win} holds.
\end{proof}

\subsection{General Mean-Payoff Taxman Games}
We extend the solution to general games. Recall that the threshold ratio in mean-payoff games is a necessary and sufficient initial ratio with which \Max can guarantee a payoff of at least $0$.
\begin{theorem}
\label{thm:general MP}
Threshold ratios exist in mean-payoff taxman games.
\end{theorem}
\begin{proof}
Consider a mean-payoff taxman game $\G = \zug{V, E, w}$ with taxman parameter $\tau$. If $\G$ is strongly-connected, then by Theorem~\ref{thm:SCC-MP}, the threshold ratio in all vertices in $\G$ is the same and is $r \in (0,1)$ for $r$ such that $\MP(\RTB^{F(\tau, r)}(\G)) = 0$. If no such $r$ exists, then either $\MP(\RTB^{F(\tau, 1)}(\G)) < 0$, in which case the threshold ratios are $1$, or $\MP(\RTB^{F(\tau, 0)}(\G)) > 0$, in which case the threshold ratios are $0$. The proof for general games follows along the same lines as the proof for reachability games. For each bottom strongly-connected component $S_i$ of $\G$ we find the threshold ratio $r_i \in (0,1)$ as in the above. We play a ``generalized'' reachability game on $\G$ as follows. The game ends once the token reaches one of the BSCCs in $\G$. \Max wins the game iff the first time the game enters a BSCC $S_i$, \Max's ratio is greater than $r_i$. Showing existence of threshold ratios in the generalized game follows the same argument as for reachability games \cite{LLPSU99}.
\end{proof}

\section{Computational Complexity}
We show, for the first time, computational complexity results for taxman games. We study the following problem, which we call THRESH: given a taxman game $\G$ with taxman parameter $\tau$ and a vertex $v_0$ in $\G$, decide whether $\thresh(v_0) \geq 0.5$. The correspondence in Theorem~\ref{thm:SCC-MP} gives the second part of the following theorem, and for the first part, in Appendix~\ref{app:complexity}, we show a reduction from THRESH to the {\em existential theory of the reals} \cite{Can88}.

\begin{theorem}
\label{thm:complexity}
For taxman reachability, parity, and mean-payoff games THRESH is in PSPACE. For strongly-connected mean-payoff games, THRESH is in NP $\cap$ coNP.
\end{theorem}

\section{Discussion}
We study, for the first time, infinite-duration taxman bidding games, which span the spectrum between Richman and poorman bidding. For qualitative objectives, we show that the properties of taxman coincide with these of Richman and poorman bidding. For mean-payoff games, where Richman and poorman bidding have an elegant though surprisingly different mathematical structure, we show a complete understanding of taxman games. Our study of mean-payoff taxman games sheds light on these differences and similarities between the two bidding rules. Unlike previous proof techniques, which were ad-hoc, we expect our technique to be easier to generalize beyond taxman games, where they can be used to introduce concepts like multi-players or partial information into bidding games.

\bibliography{../ga}

\appendix
\section{Proof of Lemma~\ref{lem:SCC-qual}}
\label{app:SCC-qual}
If \PT wins the $k$-th bidding with $1\leq k \leq n$, then the new ratio is
\begin{equation*}
\begin{split}
    \frac{\epsilon -m-mr-\dots-mr^{k-1}+\tau mr^k}{1-\epsilon+\tau m+\tau mr+\dots+\tau mr^{k-1}-mr^k} &= \frac{\epsilon -m\frac{r^k-1}{r-1}+\tau mr^k}{1-\epsilon +\tau m\frac{r^k-1}{r-1}-mr^k} \\
    &> \frac{\epsilon +mr^k(\tau - \frac{1}{r-1})}{1-\epsilon -mr^k(1-\frac{\tau}{r-1})} \\
    &> \frac{\epsilon +mr^k(\tau - \frac{1}{r-1})}{1-\epsilon} \\
    &\stackrel{\tau>\frac{2}{r-1}}{>} \frac{\epsilon}{1-\epsilon} + \frac{mr^k}{(1-\epsilon)(r-1)} \\
    &\geq \frac{\epsilon}{1-\epsilon} + \frac{mr}{(1-\epsilon)(r-1)},
\end{split}
\end{equation*}
Thus, the ratio increases by a fixed amount $b=\frac{mr}{(1-\epsilon)(r-1)}>0$. Let $\epsilon_1$ be the new (normalized) ratio of \PO. Since $0<m<\frac{\epsilon}{r^n}<\frac{\epsilon_1}{r^n}$, \PO can repeat the same process and again either win the game in at most $n$ steps or increase his budget ratio by at least $b$. Note that $\frac{mr}{(1-\epsilon)(r-1)}$ is an increasing function of $\epsilon$. Proceeding like this, eventually either \PO wins the game, or his normalized budget exceeds $1-2^{-n}$, in which case he can win $n$ consecutive biddings by bidding $2^{-n}, 2^{-n+1},\dots,2^{-1}$.

\section{Proof of Lemma~\ref{lem:ratios}}
\label{app:point 2}
We conclude by showing that Point~\ref{win} holds. Let $s \in S \setminus \set{0}$ and $1 \leq x < 1+rs$. Intuitively, if \Min wins the bidding, we reach a position that is less than $1$. We show that $1-r_x<s r(1-r)\beta_x$, therefore proving that \Min has insufficient budget to win this bid. 
Taking $(\gamma_x)_{x\geq 1}$ and $(\beta_x)_{x\geq 1}$ as in the above, we have $1-r_x=(1-\gamma^{x-1})(1-r)$ and $\beta_x=\beta\gamma^{x-1}$. Hence it suffices to prove that $\gamma^{x-1} > \frac{1}{1+sr\beta}$.
As $x-1<sr$ and $\gamma\in(0,1)$, we have $\gamma^{x-1}>\gamma^{sr}$. On the other hand, we established \eqref{eq:gamma}, thus as $[\tau+(1-\tau)r]\leq 1$ and $sr\beta\geq 0$, we conclude $\gamma^{sr} \geq \frac{1}{1+sr\beta[\tau+(1-\tau)r]} \geq \frac{1}{1+sr\beta}$.
\stam{
\begin{equation*}
    \gamma^{sr} \geq \frac{1}{1+sr\beta[\tau+(1-\tau)r]} \geq \frac{1}{1+sr\beta}.
\end{equation*}}

\section{Proof of Theorem~\ref{thm:complexity}}
\label{app:complexity}
For strongly-connected mean-payoff games, the theorem follows from Theorem~\ref{thm:SCC-MP} and the fact that solving stochastic mean-payoff games is in NP $\cap$ coNP \cite{ZP96}. In the other cases, we reduce THRESH to the existential theory of the reals, which is known to be in PSPACE \cite{Can88}. We describe the solution for reachability games and the reduction for the other objectives is similar. We start by guessing, for each $v \in V$, two vertices $v^+$ and $v^-$.  For each vertex $v$, we introduce a variable $x_v$. The constraints we use are the following. For the target $t \in V$ of \PO, we add a constraint $x_t = 0$. For every $v \in V$ from which there is no path to $t$, we add $x_v = 1$. For every other vertex $v$, we add constraints $x_{v^-} \leq x_u \leq x_{v^+}$, for every $u \in N(v)$, and a constraint $x_v = \big(x_{v^-}+x_{v^+}-\tau \cdot x_{v^-}\big)/\big(2-\tau \cdot (1+x_{v^-}-x_{v^+})\big)$. Finally, for the initial vertex $v_0$, we add a constraint $x_{v_0} \geq 1/2$. By Theorem~\ref{thm:reach}, the program has a solution iff  $\thresh(v_0) \geq 0.5$, and we are done.

\end{document}